\documentclass[conference]{IEEEtran}
\usepackage{amsmath}
\usepackage{amssymb}
\usepackage{amsthm}
\usepackage[T1]{fontenc}
\usepackage{hyperref}
\usepackage{mleftright}

\newcommand{\binaryemptyleft}{\mkern-\medmuskip{}}
\newcommand{\card}[1]{\lvert#1\rvert}
\newcommand{\infvphantomsup}{\operatorname*{inf\vphantom{sup}}}
\newcommand{\pr}[1]{\Pr[#1]}
\newcommand{\qxcommaqylessspace}{Q_X\hspace{-0.1em},Q_Y}
\newcommand{\ratepxytilde}{\tilde{\mathsf{E}}_\mathsf{P}}
\newcommand{\ratepxy}{\mathsf{E}_\mathsf{P}}
\newcommand{\rateqxqy}{\mathsf{E}_\mathsf{Q}}
\newcommand{\reals}{\mathbb{R}}
\newcommand{\set}[1]{\mathcal{#1}}
\newcommand{\spaceland}{\hspace{0.07em}\land\hspace{0.07em}}
\newcommand{\spacelor}{\hspace{0.07em}\lor\hspace{0.07em}}
\newtheorem{Corollary}{Corollary}
\newtheorem{Example}[Corollary]{Example}
\newtheorem{Lemma}[Corollary]{Lemma}
\newtheorem{Remark}[Corollary]{Remark}
\newtheorem{Theorem}[Corollary]{Theorem}

\begin{document}
\title{Testing Against Independence and\\a R\'enyi Information Measure}

\author{\IEEEauthorblockN{Amos Lapidoth and Christoph Pfister}
\IEEEauthorblockA{Signal and Information Processing Laboratory\\
ETH Zurich, 8092 Zurich, Switzerland\\
Email: \{lapidoth,pfister\}@isi.ee.ethz.ch}}

\maketitle

\begin{abstract}
The achievable error-exponent pairs for the type~I and type~II errors are characterized in a hypothesis testing setup where the observation consists of independent and identically distributed samples from either a known joint probability distribution or an unknown product distribution.
The empirical mutual information test, the Hoeffding test, and the generalized likelihood-ratio test are all shown to be asymptotically optimal.
An expression based on a R\'enyi measure of dependence is shown to be the Fenchel biconjugate of the error-exponent function obtained by fixing one error exponent and optimizing the other.
An example is provided where the error-exponent function is not convex and thus not equal to its Fenchel biconjugate.
\end{abstract}

\section{Introduction}

Let $\set{X}$ and $\set{Y}$ be finite sets and $P_{XY}$ a probability mass function (PMF) over $\set{X} \times \set{Y}$.
Based on a sequence of pairs of random variables $\{(X_i,Y_i)\}_{i=1}^n$, we want to distinguish between two hypotheses:
\begin{itemize}
\item[0)] Under the null hypothesis, $(X_1,Y_1),\ldots,(X_n,Y_n)$ are IID according to $P_{XY}$.
\item[1)] Under the alternative hypothesis, $(X_1,Y_1),\ldots,(X_n,Y_n)$ are IID according to some unknown PMF of the form $Q_{XY} = Q_X Q_Y$, where $Q_X \in \set{P}(\set{X})$ and $Q_Y \in \set{P}(\set{Y})$ are arbitrary PMFs over $\set{X}$ and $\set{Y}$, respectively.
\end{itemize}

An error-exponent pair $(\ratepxy,\rateqxqy) \in \reals^2$ is achievable if there exists a sequence of deterministic tests $\{T_n\}_{n=1}^\infty$ satisfying the following two conditions:
\begin{IEEEeqnarray}{l}
\liminf_{n \to \infty} -\frac{1}{n} \log P_{XY}^{\times n}[T_n(X^n,Y^n) = 1] > \ratepxy,\IEEEeqnarraynumspace\label{eq:deferrexponentpxy}\\
\liminf_{n \to \infty} \hspace{-0.03em}\inf_{\qxcommaqylessspace}\! -\frac{1}{n} \log (Q_X Q_Y)^{\times n}[T_n(X^n,Y^n) = 0] > \rateqxqy,\!\IEEEeqnarraynumspace\label{eq:deferrexponentqxqy}
\end{IEEEeqnarray}
where a deterministic test $T_n$ is a function from $\set{X}^n \times \set{Y}^n$ to $\{0,1\}$;
we denote by $R_{XY}^{\times n}[\set{A}]$ the probability of an event $\set{A}$ when $\{(X_i,Y_i)\}_{i=1}^n$ are IID according to $R_{XY}$;
the infimum is over all $Q_X \in \set{P}(\set{X})$ and all $Q_Y \in \set{P}(\set{Y})$; and
all logarithms in this paper are natural logarithms.
If an error-exponent pair $(\ratepxy,\rateqxqy)$ is achievable, then, since the inequalities in \eqref{eq:deferrexponentpxy} and \eqref{eq:deferrexponentqxqy} are strict, there exists a sequence of tests $\{T_n\}_{n=1}^\infty$ such that for sufficiently large $n$ and for all $(Q_X,Q_Y) \in \set{P}(\set{X}) \times \set{P}(\set{Y})$,
\begin{IEEEeqnarray}{rCl}
P_{XY}^{\times n}[T_n(X^n,Y^n) = 1] &\le& e^{-n \ratepxy},\IEEEeqnarraynumspace\label{eq:pxyerrorsmallerexp}\\
(Q_X Q_Y)^{\times n}[T_n(X^n,Y^n) = 0] &\le& e^{-n \rateqxqy}.\IEEEeqnarraynumspace\label{eq:qxqyerrorsmallerexp}
\end{IEEEeqnarray}
(The reverse is not true: \eqref{eq:pxyerrorsmallerexp} and \eqref{eq:qxqyerrorsmallerexp} are not sufficient for the achievability of the pair $(\ratepxy,\rateqxqy)$; see Section~\ref{sec:preliminaries} for more motivation for our definition.)

Our first result characterizes the achievable error-exponent pairs.

\begin{Theorem}
\label{thm:pairachievableiffrxy}
An error-exponent pair $(\ratepxy,\rateqxqy)$ is achievable if, and only if, for all $R_{XY} \in \set{P}(\set{X} \times \set{Y})$,
\begin{IEEEeqnarray}{l}
\bigl(D(R_{XY}\|P_{XY}) > \ratepxy\bigr) \spacelor \bigl(D(R_{XY}\|R_X R_Y) > \rateqxqy\bigr).\IEEEeqnarraynumspace
\end{IEEEeqnarray}
This characterization is also valid when randomized tests are allowed in \eqref{eq:deferrexponentpxy} and \eqref{eq:deferrexponentqxqy}.
\end{Theorem}

In Lemmas~\ref{lma:mutualinformationtest}--\ref{lma:glrt} we show that the empirical mutual information test, the Hoeffding test, and the generalized likelihood-ratio test can achieve every achievable error-exponent pair.
Defining the error-exponent functions $\ratepxy\colon \reals \to \reals \cup \{+\infty\}$ and $\rateqxqy\colon \reals \to \reals \cup \{+\infty\}$ as
\begin{IEEEeqnarray}{rCl}
\ratepxy(\rateqxqy) &\triangleq& \sup \, \{\ratepxy \in \reals : (\ratepxy,\rateqxqy) \text{ is achievable}\},\IEEEeqnarraynumspace\\
\rateqxqy(\ratepxy) &\triangleq& \sup \, \{\rateqxqy \in \reals : (\ratepxy,\rateqxqy) \text{ is achievable}\},\IEEEeqnarraynumspace
\end{IEEEeqnarray}
we obtain

\begin{Corollary}
\label{cor:expfixedexpoptimized}
For all $\rateqxqy \in \reals$,
\begin{IEEEeqnarray}{l}
\ratepxy(\rateqxqy) = \inf_{\substack{R_{XY} \in \set{P}(\set{X} \times \set{Y}):\\D(R_{XY}\|R_X R_Y) \le \rateqxqy}} D(R_{XY}\|P_{XY}),\IEEEeqnarraynumspace\label{eq:defepofeq}
\end{IEEEeqnarray}
and for all $\ratepxy \in \reals$,
\begin{IEEEeqnarray}{l}
\rateqxqy(\ratepxy) = \inf_{\substack{R_{XY} \in \set{P}(\set{X} \times \set{Y}):\\D(R_{XY}\|P_{XY}) \le \ratepxy}} D(R_{XY}\|R_X R_Y).\IEEEeqnarraynumspace\label{eq:defeqofep}
\end{IEEEeqnarray}
\end{Corollary}

Our next result relates the R\'enyi measure of dependence $J_\alpha(P_{XY})$ to $\ratepxy^{**}(\cdot)$, the Fenchel biconjugate of $\ratepxy(\cdot)$.
Both $J_\alpha(P_{XY})$ and $\ratepxy^{**}(\cdot)$ are discussed in Section~\ref{sec:preliminaries}.
(The analogous result for $\rateqxqy^{**}(\cdot)$ is Theorem~\ref{thm:renyiinfvsrateqxqy}.)

\begin{Theorem}
\label{thm:renyiinfvsratepxy}
For all $\rateqxqy \in \reals$,
\begin{IEEEeqnarray}{l}
\sup_{\alpha \in (0,1]} \frac{1 - \alpha}{\alpha} \bigl(J_\alpha(P_{XY}) - \rateqxqy\bigr) = \ratepxy^{**}(\rateqxqy).\IEEEeqnarraynumspace
\end{IEEEeqnarray}
Furthermore, $\ratepxy^{**}(\rateqxqy) = \ratepxy(\rateqxqy)$ for all $\rateqxqy \in \reals$ if, and only if, $\ratepxy(\cdot)$ is convex on $\reals$.
\end{Theorem}

Our last contribution is Example~\ref{exa:nonconvex}, where $\ratepxy(\cdot)$ is not convex and thus for some $\rateqxqy \in \reals$,
\begin{IEEEeqnarray}{l}
\sup_{\alpha \in (0,1]} \frac{1 - \alpha}{\alpha} \bigl(J_\alpha(P_{XY}) - \rateqxqy\bigr) \ne \ratepxy(\rateqxqy).\IEEEeqnarraynumspace\label{eq:renyineepeq}
\end{IEEEeqnarray}

The rest of this paper is organized as follows:
in Section~\ref{sec:preliminaries}, we review the R\'enyi divergence and the Fenchel conjugation;
in Section~\ref{sec:relatedwork}, we review results on simple and composite hypothesis testing;
in Section~\ref{sec:achievableexponents}, we prove Theorem~\ref{thm:pairachievableiffrxy} and provide asymptotically optimal tests;
in Section~\ref{sec:errorexponentvsrenyi}, we relate $J_\alpha(P_{XY})$ to the Fenchel biconjugates $\ratepxy^{**}$ and $\rateqxqy^{**}$; and
in Section~\ref{sec:examplenotconvex}, we discuss Example~\ref{exa:nonconvex}, where $\ratepxy(\cdot)$ is not convex.
Additional proofs can be found in the \hyperref[sec:appendix]{Appendix}.

\section{Preliminaries}
\label{sec:preliminaries}

Let $P$ and $Q$ be PMFs over a finite set $\set{Z}$.
The relative entropy (or Kullback--Leibler divergence) is defined as
\begin{IEEEeqnarray}{l}
D(P\|Q) \triangleq \sum_{z \in \set{Z}} P(z) \log \frac{P(z)}{Q(z)}\IEEEeqnarraynumspace
\end{IEEEeqnarray}
with the conventions that $0 \log (0/q) = 0$ for all $q \ge 0$ and $p \log (p/0) = +\infty$ for all $p > 0$.
The R\'enyi divergence of order $\alpha$ \cite{RenyiDivergence,VanErvenHarremoes} is defined for all positive $\alpha$ other than 1 as
\begin{IEEEeqnarray}{l}
D_\alpha(P\|Q) \triangleq \frac{1}{\alpha - 1} \log \sum_{z \in \set{Z}} P(z)^\alpha \, Q(z)^{1 - \alpha}\IEEEeqnarraynumspace
\end{IEEEeqnarray}
with the conventions that $\log 0 = -\infty$ and that for $\alpha > 1$, we read $P(z)^\alpha \, Q(z)^{1 - \alpha}$ as $P(z)^\alpha / Q(z)^{\alpha - 1}$ and use $0/0 = 0$ and $p / 0 = +\infty$ for all $p > 0$.
By continuous extension \cite[Theorem~5]{VanErvenHarremoes}, $D_1(P\|Q) \triangleq D(P\|Q)$.

The measure of dependence $J_\alpha(P_{XY})$ \cite{TwoMeasuresOfDependence} is defined as
\begin{IEEEeqnarray}{l}
J_\alpha(P_{XY}) \triangleq \min_{Q_X \in \set{P}(\set{X}),\, Q_Y \in \set{P}(\set{Y})} D_\alpha(P_{XY}\|Q_X Q_Y)\IEEEeqnarraynumspace
\end{IEEEeqnarray}
for all positive $\alpha$ and as zero when $\alpha$ is zero.

The convex conjugate (or Fenchel conjugate) of a function $f\colon \reals \to [-\infty,+\infty]$ is the function $f^*\colon \reals \to [-\infty,+\infty]$,
\begin{IEEEeqnarray}{l}
f^*(\lambda) \triangleq \sup_{x \in \reals} \, [\lambda \, x - f(x)].\IEEEeqnarraynumspace
\end{IEEEeqnarray}
It is lower semicontinuous and convex \cite[Section~7.1 and Proposition~1.2.2]{Bertsekas}.

The (Fenchel) biconjugate of a function $f\colon \reals \to [-\infty,+\infty]$ is $f^{**}\colon \reals \to [-\infty,+\infty]$, the convex conjugate of $f^*$.
For every $f$ and for every $x \in \reals$, $f^{**}(x) \le f(x)$ \cite[Section~4.2]{BorweinLewis}.

We next motivate the strict inequalities in \eqref{eq:deferrexponentpxy} and \eqref{eq:deferrexponentqxqy}.
Let $\ratepxytilde(\rateqxqy)$ denote the error-exponent function that would have resulted had we replaced the strict inequalities in \eqref{eq:deferrexponentpxy} and \eqref{eq:deferrexponentqxqy} with weak inequalities.
Then, $\ratepxytilde(\cdot)$ and $\ratepxytilde^{**}(\cdot)$ cannot be equal because, unlike $\ratepxy(\cdot)$, $\ratepxytilde(\cdot)$ is not lower semicontinuous.
The difference between $\ratepxy(\cdot)$ and $\ratepxytilde(\cdot)$ is best seen at zero:
While $\ratepxytilde(0)$ is $+\infty$, it turns out that $\ratepxy(0)$ is the optimal error exponent if for a fixed $\epsilon \in (0,1)$, we require the tests to satisfy $(Q_X Q_Y)^{\times n}[T_n(X^n,Y^n) = 0] \le \epsilon$ for all $n$ and all $(Q_X,Q_Y)$.
(This setup is similar to the one in Stein's lemma \cite[Corollary~1.2]{CsiszarKorner}; we do not explore it further in this paper.)

To see that \eqref{eq:pxyerrorsmallerexp} and \eqref{eq:qxqyerrorsmallerexp} are not sufficient for the achievability of an error-exponent pair, observe that \eqref{eq:pxyerrorsmallerexp} and \eqref{eq:qxqyerrorsmallerexp} hold for every $\ratepxy \in \reals$ if $\rateqxqy = 0$ and $T_n(X^n,Y^n) = 0$ irrespective of $X^n$ and $Y^n$.
Yet, \eqref{eq:defepofeq} implies that $\ratepxy(0)$ is finite, so $(\ratepxy,0)$ is not achievable for every $\ratepxy$.

We conclude this section with two lemmas.

\begin{Lemma}
\label{lma:drxyrxryledrxyqxqy}
For all $R_{XY}$, $Q_X$, and $Q_Y$,
\begin{IEEEeqnarray}{l}
D(R_{XY}\|Q_X Q_Y) \ge D(R_{XY}\|R_X R_Y).\IEEEeqnarraynumspace
\end{IEEEeqnarray}
Consequently,
\begin{IEEEeqnarray}{l}
\inf_{\qxcommaqylessspace} D(R_{XY}\|Q_X Q_Y) = D(R_{XY}\|R_X R_Y).\IEEEeqnarraynumspace\label{eq:infdrxyqxqy}
\end{IEEEeqnarray}
\end{Lemma}

\begin{proof}
We have
\begin{IEEEeqnarray}{rCl}
\IEEEeqnarraymulticol{3}{l}{D(R_{XY}\|Q_X Q_Y)}\IEEEeqnarraynumspace\nonumber\\*\quad
&=& D(R_{XY}\|R_X R_Y) + D(R_X\|Q_X) + D(R_Y\|Q_Y)\IEEEeqnarraynumspace\\
&\ge& D(R_{XY}\|R_X R_Y),\IEEEeqnarraynumspace\label{eq:drxyrxryledrxyqxqy}
\end{IEEEeqnarray}
where \eqref{eq:drxyrxryledrxyqxqy} holds because $D(P\|Q) \ge 0$.
Equality is achieved for $Q_X = R_X$ and $Q_Y = R_Y$, which proves \eqref{eq:infdrxyqxqy}.
\end{proof}

\begin{Lemma}
\label{lma:infdrqdrpeqsupalpha}
Let $P$ and $Q$ be PMFs over a finite set $\set{Z}$.
Then, for all $\rateqxqy \in \reals$,
\begin{IEEEeqnarray}{l}
\inf_{\substack{R \in \set{P}(\set{Z}):\\D(R\|Q) \le \rateqxqy}} D(R\|P) = \sup_{\alpha \in (0,1]} \frac{1 - \alpha}{\alpha} \bigl(D_\alpha(P\|Q) - \rateqxqy\bigr).\IEEEeqnarraynumspace
\end{IEEEeqnarray}
\end{Lemma}

\begin{proof}
Omitted.
\end{proof}

\section{Related Work}
\label{sec:relatedwork}

Let $P$ and $Q$ be PMFs over a finite set $\set{Z}$.
In the simple hypothesis testing setup where one has to guess whether $\{Z_i\}_{i=1}^n$ are IID according to $P$ or $Q$, Hoeffding \cite{Hoeffding} and Csisz\'ar and Longo \cite{CsiszarLongo} essentially showed that
\begin{IEEEeqnarray}{l}
\ratepxytilde(\rateqxqy) = \inf_{R \in \set{P}(\set{Z}):\,D(R\|Q) \le \rateqxqy} D(R\|P),\IEEEeqnarraynumspace
\end{IEEEeqnarray}
where $\ratepxytilde(\cdot)$ is the error-exponent function for the simple hypothesis testing setup.
More properties of $\ratepxytilde(\cdot)$ were studied by Blahut \cite{Blahut}; relevant for us is
\begin{IEEEeqnarray}{l}
\ratepxytilde(\rateqxqy) = \sup_{\alpha \in (0,1]} \frac{1 - \alpha}{\alpha} \bigl(D_\alpha(P\|Q) - \rateqxqy\bigr),\IEEEeqnarraynumspace
\end{IEEEeqnarray}
which follows from \cite[Theorem~7]{Blahut} by substituting $\alpha = \smash{\frac{1}{1 + s}}$ and identifying the R\'enyi divergence.

In the composite hypothesis testing setup where $P$ is tested against an unknown $Q$ from some set $\set{Q}$, Hoeffding \cite{Hoeffding} showed that his likelihood-ratio test is asymptotically optimal against all $Q \in \set{Q}$; see also \cite[Problem~2.13(b)]{CsiszarKorner}.
This test statistic is used in Lemma~\ref{lma:hoeffdingtest}.

For the hypothesis testing setup of this paper, Tomamichel and Hayashi \cite[first part of (57)]{TomamichelHayashi} showed that for sufficiently large $\rateqxqy$,
\begin{IEEEeqnarray}{l}
\sup_{\alpha \in (\frac{1}{2},1)} \frac{1 - \alpha}{\alpha} \bigl(J_\alpha(P_{XY}) - \rateqxqy\bigr) = \ratepxy(\rateqxqy).\IEEEeqnarraynumspace\label{eq:tomamichelhayashiepeq}
\end{IEEEeqnarray}
We provide a negative answer to the question at the end of the paragraph in \cite{TomamichelHayashi}: an equality of the form \eqref{eq:tomamichelhayashiepeq} does not hold in general because the LHS of \eqref{eq:tomamichelhayashiepeq} is always convex in $\rateqxqy$, but $\ratepxy(\cdot)$ from Example~\ref{exa:nonconvex} is not convex.

Conditions for which the generalized likelihood-ratio test is asymptotically optimal in a Neyman--Pearson sense are studied in \cite{GLRT}.
A different approach to composite hypothesis testing has been proposed in \cite{CompositeHypothesisMinimax}.

Independence testing is a related setup where one wants to know whether or not the PMF generating $\{(X_i,Y_i)\}_{i=1}^n$ has a product form (whereas here, we test a fixed $P_{XY}$ against an unknown product distribution).
Since the empirical mutual information in Lemma~\ref{lma:mutualinformationtest} does not depend on $P_{XY}$, it can also be used for independence testing; see for example \cite[``$G$-test of independence'']{SokalRohlf}, where $G$ is $2n$ times the empirical mutual information.

\section{Achievable Error-Exponent Pairs}
\label{sec:achievableexponents}

After two preparatory lemmas, we present in Lemmas~\ref{lma:mutualinformationtest}--\ref{lma:glrt} three tests that achieve any error-exponent pair $(\ratepxy,\rateqxqy)$ for which
\begin{IEEEeqnarray}{l}
\bigl(D(R_{XY}\|P_{XY}) > \ratepxy\bigr) \spacelor \bigl(D(R_{XY}\|R_X R_Y) > \rateqxqy\bigr)\IEEEeqnarraynumspace\label{eq:pairachievrxy}
\end{IEEEeqnarray}
holds for all $R_{XY} \in \set{P}(\set{X} \times \set{Y})$.
These tests are all based on the type \cite{CsiszarKorner} $\hat{R}_{XY}$ of the sequence $\{(X_i,Y_i)\}_{i=1}^n$.
The asymptotic optimality of these tests follows from the converse proved in Lemma~\ref{lma:exponentnotachievable}, which establishes Theorem~\ref{thm:pairachievableiffrxy} and Corollary~\ref{cor:expfixedexpoptimized}.

\begin{Lemma}
\label{lma:rxystrictimpliesepsilon}
If \eqref{eq:pairachievrxy} holds for all $R_{XY} \in \set{P}(\set{X} \times \set{Y})$, then there exists an $\epsilon > 0$ such that for all $R_{XY} \in \set{P}(\set{X} \times \set{Y})$,
\begin{IEEEeqnarray}{l}
\bigl(D(R_{XY}\|P_{XY}) \ge \ratepxy + \epsilon\bigr) \spacelor \bigl(D(R_{XY}\|R_X R_Y) \ge \rateqxqy + \epsilon\bigr).\nonumber\\*\IEEEeqnarraynumspace\label{eq:lmarxyepsepsilon}
\end{IEEEeqnarray}
\end{Lemma}

\begin{proof}
Define the function $f\colon \set{P}(\set{X} \times \set{Y}) \to \reals \cup \{+\infty\}$,
\begin{IEEEeqnarray}{rl}
R_{XY} \mapsto \max \bigl\{&D(R_{XY}\|P_{XY}) - \ratepxy,\IEEEeqnarraynumspace\nonumber\\*
&D(R_{XY}\|R_X R_Y) - \rateqxqy\bigr\}.\IEEEeqnarraynumspace
\end{IEEEeqnarray}
Suppose that \eqref{eq:pairachievrxy} holds for all $R_{XY} \in \set{P}(\set{X} \times \set{Y})$, and consider
\begin{IEEEeqnarray}{l}
\eta \triangleq \inf_{R_{XY} \in \set{P}(\set{X} \times \set{Y})} f(R_{XY}).\IEEEeqnarraynumspace
\end{IEEEeqnarray}
If $\eta > 0$, then \eqref{eq:lmarxyepsepsilon} holds with $\epsilon = \eta$.
We show by contradiction that $\eta \le 0$ is impossible.
Assume $\eta \le 0$.
Observe that $f$ is lower semicontinuous on $\set{P}(\set{X} \times \set{Y})$ and that $\set{P}(\set{X} \times \set{Y})$ is a compact set.
By the extreme value theorem, there would exist an $R_{XY}^* \in \set{P}(\set{X} \times \set{Y})$ with $f(R_{XY}^*) = \eta \le 0$.
This leads to a contradiction because then \eqref{eq:pairachievrxy} would not hold for $R_{XY}^*$.
\end{proof}

\begin{Lemma}
\label{lma:achievablesanov}
Let $\ratepxy$, $\rateqxqy$, and $\epsilon > 0$ be such that \eqref{eq:lmarxyepsepsilon} holds for all $R_{XY} \in \set{P}(\set{X} \times \set{Y})$.
Define $\tau \triangleq (n + 1)^{\card{\set{X} \times \set{Y}}}$.
Then, for all $Q_X \in \set{P}(\set{X})$ and all $Q_Y \in \set{P}(\set{Y})$,
\begin{IEEEeqnarray}{rCl}
P_{XY}^{\times n}[D(\hat{R}_{XY}\|P_{XY}) \ge \ratepxy + \epsilon] &\le& \tau \, e^{-n (\ratepxy + \epsilon)},\hspace{-0.09em}\IEEEeqnarraynumspace\label{eq:achsanova}\\
P_{XY}^{\times n}[D(\hat{R}_{XY}\|\hat{R}_X \hat{R}_Y) < \rateqxqy + \epsilon] &\le& \tau \, e^{-n (\ratepxy + \epsilon)},\hspace{-0.09em}\IEEEeqnarraynumspace\label{eq:achsanovb}\\
(Q_X Q_Y)^{\times n}[D(\hat{R}_{XY}\|P_{XY}) < \ratepxy + \epsilon] &\le& \tau \, e^{-n (\rateqxqy + \epsilon)},\hspace{-0.09em}\IEEEeqnarraynumspace\label{eq:achsanovc}\\
(Q_X Q_Y)^{\times n}[D(\hat{R}_{XY}\|\hat{R}_X \hat{R}_Y) \ge \rateqxqy + \epsilon] &\le& \tau \, e^{-n (\rateqxqy + \epsilon)}.\hspace{-0.09em}\IEEEeqnarraynumspace\label{eq:achsanovd}
\end{IEEEeqnarray}
\end{Lemma}

\begin{proof}
In the \hyperlink{prf:achievablesanov}{Appendix}.
\end{proof}

\begin{Lemma}[Empirical Mutual Information Test]
\label{lma:mutualinformationtest}
If \eqref{eq:pairachievrxy} is satisfied for all $R_{XY} \in \set{P}(\set{X} \times \set{Y})$, then there exists an $\epsilon > 0$ such that the error-exponent pair $(\ratepxy,\rateqxqy)$ is achieved by the sequence of tests
\begin{IEEEeqnarray}{l}
T_n(\hat{R}_{XY}) \triangleq \begin{cases}1 & \text{if } D(\hat{R}_{XY}\|\hat{R}_X \hat{R}_Y) < \rateqxqy + \epsilon,\\
0 & \text{otherwise.}\end{cases}\IEEEeqnarraynumspace
\end{IEEEeqnarray}
\end{Lemma}

\begin{proof}
Use the $\epsilon > 0$ from Lemma~\ref{lma:rxystrictimpliesepsilon}.
Then, the sequence of tests $\{T_n\}_{n=1}^\infty$ satisfies \eqref{eq:deferrexponentpxy} because
\begin{IEEEeqnarray}{rCl}
\IEEEeqnarraymulticol{3}{l}{P_{XY}^{\times n}[T_n(X^n,Y^n) = 1]}\IEEEeqnarraynumspace\nonumber\\*\qquad
&=& P_{XY}^{\times n}[D(\hat{R}_{XY}\|\hat{R}_X \hat{R}_Y) < \rateqxqy + \epsilon]\IEEEeqnarraynumspace\\
&\le& (n + 1)^{\card{\set{X} \times \set{Y}}} \cdot e^{-n (\ratepxy + \epsilon)},\IEEEeqnarraynumspace\label{eq:mitestb}
\end{IEEEeqnarray}
where \eqref{eq:mitestb} follows from Lemma~\ref{lma:achievablesanov}.
Similarly, the sequence of tests $\{T_n\}_{n=1}^\infty$ satisfies \eqref{eq:deferrexponentqxqy}.
\end{proof}

\begin{Lemma}[Hoeffding's Test \cite{Hoeffding}]
\label{lma:hoeffdingtest}
If \eqref{eq:pairachievrxy} is satisfied for all $R_{XY} \in \set{P}(\set{X} \times \set{Y})$, then there exists an $\epsilon > 0$ such that the error-exponent pair $(\ratepxy,\rateqxqy)$ is achieved by the sequence of tests
\begin{IEEEeqnarray}{l}
T_n(\hat{R}_{XY}) \triangleq \begin{cases}0 & \text{if } D(\hat{R}_{XY}\|P_{XY}) < \ratepxy + \epsilon,\\
1 & \text{otherwise.}\end{cases}\IEEEeqnarraynumspace
\end{IEEEeqnarray}
\end{Lemma}

\begin{proof}
The proof is very similar to the proof of Lemma~\ref{lma:mutualinformationtest}.
\end{proof}

\begin{Lemma}[Generalized Likelihood-Ratio Test]
\label{lma:glrt}
The logarithm of the generalized likelihood ratio, divided by $n$, is
\begin{IEEEeqnarray}{rCl}
\Gamma &\triangleq& \frac{1}{n} \log \frac{P_{XY}^{\times n}(X^n, Y^n)}{\sup\limits_{Q_X \in \set{P}(\set{X}),\, Q_Y \in \set{P}(\set{Y})} (Q_X Q_Y)^{\times n}(X^n, Y^n)}\IEEEeqnarraynumspace\\
&=& D(\hat{R}_{XY}\|\hat{R}_X \hat{R}_Y) - D(\hat{R}_{XY}\|P_{XY}).\IEEEeqnarraynumspace\label{eq:qlrtb}
\end{IEEEeqnarray}
If \eqref{eq:pairachievrxy} is satisfied for all $R_{XY} \in \set{P}(\set{X} \times \set{Y})$, then the error-exponent pair $(\ratepxy,\rateqxqy)$ is achieved by the sequence of tests
\begin{IEEEeqnarray}{l}
T_n(\hat{R}_{XY}) \triangleq \begin{cases}1 & \text{if } \Gamma \le \rateqxqy - \ratepxy,\\
0 & \text{otherwise.}\end{cases}\IEEEeqnarraynumspace
\end{IEEEeqnarray}
\end{Lemma}

\begin{proof}
The proof of \eqref{eq:qlrtb} is omitted.
Using the $\epsilon > 0$ from Lemma~\ref{lma:rxystrictimpliesepsilon}, the sequence of tests $\{T_n\}_{n=1}^\infty$ satisfies \eqref{eq:deferrexponentpxy} because
\begin{IEEEeqnarray}{rCl}
\IEEEeqnarraymulticol{3}{l}{P_{XY}^{\times n}[T_n(X^n,Y^n) = 1]}\IEEEeqnarraynumspace\nonumber\\*\qquad
&\le& P_{XY}^{\times n}[D(\hat{R}_{XY}\|\hat{R}_X \hat{R}_Y) < \rateqxqy + \epsilon]\IEEEeqnarraynumspace\nonumber\\*
&& \binaryemptyleft + P_{XY}^{\times n}[D(\hat{R}_{XY}\|P_{XY}) \ge \ratepxy + \epsilon]\IEEEeqnarraynumspace\label{eq:glrtpxybounda}\\
&\le& 2\,(n + 1)^{\card{\set{X} \times \set{Y}}} \cdot e^{-n (\ratepxy + \epsilon)},\IEEEeqnarraynumspace\label{eq:glrtpxyboundb}
\end{IEEEeqnarray}
where \eqref{eq:glrtpxybounda} follows from the union bound because the events $D(\hat{R}_{XY}\|\hat{R}_X \hat{R}_Y) \ge \rateqxqy + \epsilon$ and $D(\hat{R}_{XY}\|P_{XY}) < \ratepxy + \epsilon$ imply $\Gamma > \rateqxqy - \ratepxy$;
and \eqref{eq:glrtpxyboundb} follows from Lemma~\ref{lma:achievablesanov}.
In the same way, the sequence of tests $\{T_n\}_{n=1}^\infty$ satisfies \eqref{eq:deferrexponentqxqy}.
\end{proof}

\begin{Lemma}
\label{lma:exponentnotachievable}
If \eqref{eq:pairachievrxy} does not hold for all $R_{XY} \in \set{P}(\set{X} \times \set{Y})$, i.e., if there exists an $R_{XY}^* \in \set{P}(\set{X} \times \set{Y})$ satisfying
\begin{IEEEeqnarray}{l}
\bigl(D(R_{XY}^*\|P_{XY}) \le \ratepxy\bigr) \spaceland \bigl(D(R_{XY}^*\|R_X^* R_Y^*) \le \rateqxqy\bigr),\IEEEeqnarraynumspace\label{eq:rxystarclosetopxyandrxry}
\end{IEEEeqnarray}
then the error-exponent pair $(\ratepxy,\rateqxqy)$ is not achievable.
(Not even if randomized tests are allowed.)
\end{Lemma}

\begin{proof}
In the \hyperlink{prf:exponentnotachievable}{Appendix}.
\end{proof}

\begin{proof}[Proof of Theorem~\ref{thm:pairachievableiffrxy}]
The theorem follows from Lemma~\ref{lma:mutualinformationtest} and from Lemma~\ref{lma:exponentnotachievable}.
\end{proof}

\begin{proof}[Proof of Corollary~\ref{cor:expfixedexpoptimized}]
For a fixed $\rateqxqy \in \reals$, define
\begin{IEEEeqnarray}{l}
C \triangleq \inf_{\substack{R_{XY} \in \set{P}(\set{X} \times \set{Y}):\\D(R_{XY}\|R_X R_Y) \le \rateqxqy}} D(R_{XY}\|P_{XY}).\IEEEeqnarraynumspace
\end{IEEEeqnarray}
By Theorem~\ref{thm:pairachievableiffrxy}, all error-exponent pairs $(\ratepxy, \rateqxqy)$ with $\ratepxy < C$ are achievable, while those with $\ratepxy > C$ are not.
Therefore, $\ratepxy(\rateqxqy) = C$.
An analogous argument proves \eqref{eq:defeqofep}.
\end{proof}

\section{Error-Exponent Functions and $J_\alpha(P_{XY})$}
\label{sec:errorexponentvsrenyi}

After a preparatory lemma, we prove Theorem~\ref{thm:renyiinfvsratepxy} and state Theorem~\ref{thm:renyiinfvsrateqxqy}, the analog of Theorem~\ref{thm:renyiinfvsratepxy} for $\rateqxqy^{**}(\cdot)$.

\begin{Lemma}
\label{lma:epconjugate}
The convex conjugate of $\ratepxy(\cdot)$ is
\begin{IEEEeqnarray}{l}
\ratepxy^*(\lambda) = \begin{cases}+\infty & \text{if } \lambda > 0,\\
\lambda \, J_\frac{1}{1 - \lambda}(P_{XY}) & \text{otherwise.}\end{cases}\IEEEeqnarraynumspace
\end{IEEEeqnarray}
\end{Lemma}

\begin{proof}
By the definition of the convex conjugate,
\begin{IEEEeqnarray}{l}
\ratepxy^*(\lambda) = \sup_{\rateqxqy \in \reals} [\lambda \, \rateqxqy - \ratepxy(\rateqxqy)].\IEEEeqnarraynumspace\label{eq:epconjugate}
\end{IEEEeqnarray}
For $\lambda > 0$, the RHS of \eqref{eq:epconjugate} is $+\infty$, since we can lower-bound the supremum over $\rateqxqy$ with the limit as $\rateqxqy$ tends to infinity and since $\ratepxy(\rateqxqy)$ is zero for all $\rateqxqy \ge D(P_{XY}\|P_X P_Y)$, which can be verified by choosing $R_{XY} = P_{XY}$ in the RHS of \eqref{eq:defepofeq}.

Now assume $\lambda \le 0$.
Then,
\begin{IEEEeqnarray}{rCl}
\IEEEeqnarraymulticol{3}{l}{\sup_{\rateqxqy \in \reals} [\lambda \, \rateqxqy - \ratepxy(\rateqxqy)]}\IEEEeqnarraynumspace\nonumber\\*[-0.5ex]\quad
&=& \sup_{\rateqxqy \in \reals} \biggl[\lambda \, \rateqxqy - \inf_{\substack{R_{XY} \in \set{P}(\set{X} \times \set{Y}):\\D(R_{XY}\|R_X R_Y) \le \rateqxqy}} D(R_{XY}\|P_{XY})\biggr]\IEEEeqnarraynumspace\label{eq:epconjugatea}\\
&=& \sup_{\rateqxqy \in \reals} \, \sup_{\substack{R_{XY} \in \set{P}(\set{X} \times \set{Y}):\\D(R_{XY}\|R_X R_Y) \le \rateqxqy}} [\lambda \, \rateqxqy - D(R_{XY}\|P_{XY})]\IEEEeqnarraynumspace\\
&=& \sup_{R_{XY}} \, \sup_{\rateqxqy : \, \rateqxqy \ge D(R_{XY}\|R_X R_Y)} [\lambda \, \rateqxqy - D(R_{XY}\|P_{XY})]\IEEEeqnarraynumspace\\
&=& \sup_{R_{XY}} \, [\lambda \, D(R_{XY}\|R_X R_Y) - D(R_{XY}\|P_{XY})]\IEEEeqnarraynumspace\label{eq:epconjugated}\\[-0.5ex]
&=& -(1 - \lambda) \inf_{\qxcommaqylessspace} \inf_{R_{XY}} \biggl[\frac{-\lambda}{1 - \lambda} \, D(R_{XY}\|Q_X Q_Y)\IEEEeqnarraynumspace\nonumber\\*[-0.5ex]
&& \hphantom{-(1 - \lambda) \inf_{\qxcommaqylessspace} \inf_{R_{XY}} \biggl[\binaryemptyleft} + \frac{1}{1 - \lambda} \, D(R_{XY}\|P_{XY})\biggr]\IEEEeqnarraynumspace\label{eq:epconjugatee}\\
&=& \lambda \inf_{\qxcommaqylessspace} D_\frac{1}{1 - \lambda}(P_{XY}\|Q_X Q_Y)\IEEEeqnarraynumspace\label{eq:epconjugatef}\\
&=& \lambda \, J_\frac{1}{1 - \lambda}(P_{XY}),\IEEEeqnarraynumspace\label{eq:epconjugateg}
\end{IEEEeqnarray}
where \eqref{eq:epconjugatea} follows from \eqref{eq:defepofeq};
\eqref{eq:epconjugated} holds because $\lambda \le 0$, so $\rateqxqy = D(R_{XY}\|R_X R_Y)$ achieves the maximum;
\eqref{eq:epconjugatee} follows from Lemma~\ref{lma:drxyrxryledrxyqxqy} because $\frac{-\lambda}{1 - \lambda} \ge 0$ and $1 - \lambda \ge 1$;
\eqref{eq:epconjugatef} follows from \cite[Theorem~30]{VanErvenHarremoes} with $\alpha = \frac{1}{1 - \lambda} \in (0,1]$; and
\eqref{eq:epconjugateg} follows from the definition of $J_\alpha(P_{XY})$.
(Technically, the case $\alpha = 1$ is not covered by \cite[Theorem~30]{VanErvenHarremoes}, but it is easy to see that \eqref{eq:epconjugatef} also holds if $\alpha = 1$, i.e., if $\lambda = 0$.)
\end{proof}

\begin{proof}[Proof of Theorem~\ref{thm:renyiinfvsratepxy}]
Using Lemma~\ref{lma:epconjugate}, we have
\begin{IEEEeqnarray}{rCl}
\ratepxy^{**}(\rateqxqy) &=& \sup_{\lambda \in \reals} \, [\lambda \, \rateqxqy - \ratepxy^*(\lambda)]\IEEEeqnarraynumspace\\
&=& \sup_{\lambda \le 0} \, [\lambda \, \rateqxqy - \ratepxy^*(\lambda)]\IEEEeqnarraynumspace\label{eq:biconjb}\\
&=& \sup_{\lambda \le 0} \, [\lambda \, \rateqxqy - \lambda \, J_\frac{1}{1 - \lambda}(P_{XY})]\IEEEeqnarraynumspace\\[-0.5ex]
&=& \sup_{\alpha \in (0,1]} \frac{1 - \alpha}{\alpha} \bigl(J_\alpha(P_{XY}) - \rateqxqy\bigr),\IEEEeqnarraynumspace\label{eq:biconjd}
\end{IEEEeqnarray}
where \eqref{eq:biconjb} holds because $\ratepxy^*(\lambda) = +\infty$ for all $\lambda > 0$, and
\eqref{eq:biconjd} follows from the substitution $\alpha = \frac{1}{1 - \lambda} \in (0,1]$.

By \cite[Theorem~4.2.1]{BorweinLewis}, a function $h\colon \reals \to \reals \cup \{+\infty\}$ is equal to its biconjugate if, and only if, it is lower semicontinuous and convex.
The function $\ratepxy\colon \reals \to [0,+\infty]$ is always lower semicontinuous.
(This follows from a topological argument, which is omitted here.)
Thus, $\ratepxy(\cdot)$ is equal to its biconjugate if, and only if, it is convex.
\end{proof}

\begin{Theorem}
\label{thm:renyiinfvsrateqxqy}
For all $\ratepxy \in \reals$,
\begin{IEEEeqnarray}{l}
\sup_{\alpha \in [0,1)} \biggl[J_\alpha(P_{XY}) - \frac{\alpha}{1 - \alpha} \, \ratepxy\biggr] = \rateqxqy^{**}(\ratepxy).\IEEEeqnarraynumspace
\end{IEEEeqnarray}
Furthermore, $\rateqxqy^{**}(\ratepxy) = \rateqxqy(\ratepxy)$ for all $\ratepxy \in \reals$ if, and only if, $\rateqxqy(\cdot)$ is convex on $\reals$.
\end{Theorem}

\begin{proof}
Omitted; the proof is similar to the proofs of Lemma~\ref{lma:epconjugate} and Theorem~\ref{thm:renyiinfvsratepxy}.
\end{proof}

\section{An Example Where $\ratepxy(\cdot)$ Is Not Convex}
\label{sec:examplenotconvex}

\begin{Example}
\label{exa:nonconvex}
Consider $\set{X} = \set{Y} = \{1,2,3\}$ and $P_{XY}$ given by
\begin{center}
\begin{tabular}{c|ccc}
$P_{XY}(x,y)$ & $y=1$ & $y=2$ & $y=3$\\
\hline
\\[-2.4ex]
$x=1$ & $10^{-4}$ & $\gamma$ & $\gamma$\\
$x=2$ & $\gamma$ & $10^{-4}$ & $\gamma$\\
$x=3$ & $\gamma$ & $\gamma$ & $10^{-4}\rlap{$,$}$
\end{tabular}
\end{center}
where $\gamma = \frac{9997}{60000} \approx 0.167$.
Then,
\begin{IEEEeqnarray}{l}
\ratepxy\bigl(3898 \,/\, 2^{17}\bigr) \le 58593464420737815 \,/\, 2^{56},\IEEEeqnarraynumspace\label{eq:nonconvexuba}\\
\ratepxy\bigl(3984 \,/\, 2^{17}\bigr) \le 58382556630811219 \,/\, 2^{56},\IEEEeqnarraynumspace\label{eq:nonconvexubb}\\
\ratepxy\bigl(3941 \,/\, 2^{17}\bigr) \ge 58488010525784883 \,/\, 2^{56}.\IEEEeqnarraynumspace\label{eq:nonconvexubc}
\end{IEEEeqnarray}
This implies
\begin{IEEEeqnarray}{rCl}
\IEEEeqnarraymulticol{3}{l}{\ratepxy\mleft(\frac{3898 + 3984}{2 \cdot 2^{17}}\mright) - \frac{1}{2} \, \ratepxy\mleft(\frac{3898}{2^{17}}\mright) - \frac{1}{2} \, \ratepxy\mleft(\frac{3984}{2^{17}}\mright)}\IEEEeqnarraynumspace\nonumber\\*[0.4ex]\quad
&=& \ratepxy\mleft(\frac{3941}{2^{17}}\mright) - \frac{1}{2} \, \ratepxy\mleft(\frac{3898}{2^{17}}\mright) - \frac{1}{2} \, \ratepxy\mleft(\frac{3984}{2^{17}}\mright)\IEEEeqnarraynumspace\label{eq:nonconvexgapa}\\[0.3ex]
&\ge& 10366 \,/\, 2^{56} \approx 1.44 \cdot 10^{-13},\IEEEeqnarraynumspace
\end{IEEEeqnarray}
so $\ratepxy(\cdot)$ is not convex.
(We estimate the LHS of \eqref{eq:nonconvexgapa} to be in the order of $10^{-7}$.)
\end{Example}

To verify \eqref{eq:nonconvexuba}, we use \eqref{eq:defepofeq} and check (see Remark~\ref{rmk:epeqlowerbound} below) that a specific $R_{XY} \in \set{P}(\set{X} \times \set{Y})$ satisfies
\begin{IEEEeqnarray}{rCl}
D(R_{XY}\|R_X R_Y) &\le& 3898 \,/\, 2^{17},\IEEEeqnarraynumspace\\
D(R_{XY}\|P_{XY}) &\le& 58593464420737815 \,/\, 2^{56}.\IEEEeqnarraynumspace
\end{IEEEeqnarray}
Similarly, \eqref{eq:nonconvexubb} can be verified.
Establishing \eqref{eq:nonconvexubc} is much more involved and is the topic of the rest of this section.

Let $\set{Q}$ denote the set of all product distributions on $\set{X} \times \set{Y}$,
\begin{IEEEeqnarray}{l}
\set{Q} \triangleq \{Q_{XY} \in \set{P}(\set{X} \times \set{Y}) : Q_{XY} = Q_X Q_Y\}.\IEEEeqnarraynumspace
\end{IEEEeqnarray}
We express $\set{Q}$ as a finite union, i.e.,
\begin{IEEEeqnarray}{l}
\set{Q} = \bigcup_{i=1}^{k} \set{Q}_i,\IEEEeqnarraynumspace\label{eq:qunionqi}
\end{IEEEeqnarray}
where for each $i \in \{1,\ldots,k\}$,
\begin{IEEEeqnarray}{rClClCl}
\set{Q}_i &\triangleq& \IEEEeqnarraymulticol{5}{l}{\bigl\{Q_X Q_Y : \bigl(Q_X \in \set{Q}_{X,i}\bigr) \spaceland \bigl(Q_Y \in \set{Q}_{Y,i}\bigr)\bigr\},}\IEEEeqnarraynumspace\label{eq:qxyidef}\\
\set{Q}_{X,i} &\triangleq& \bigl\{Q_X \in \set{P}(\set{X}) : \hphantom{\spaceland} \bigl(l_{i,1} &\le& Q_X(1) &\le& u_{i,1}\bigr)\IEEEeqnarraynumspace\nonumber\\*
&& \hfill \spaceland \bigl(l_{i,2} &\le& Q_X(2) &\le& u_{i,2}\bigr)\IEEEeqnarraynumspace\nonumber\\*
&& \hfill \spaceland \bigl(l_{i,3} &\le& Q_X(3) &\le& u_{i,3}\bigr)\bigr\},\IEEEeqnarraynumspace\\
\set{Q}_{Y,i} &\triangleq& \bigl\{Q_Y \in \set{P}(\set{Y}) : \hphantom{\spaceland} \hfill \bigl(l_{i,4} &\le& Q_Y(1) &\le& u_{i,4}\bigr)\IEEEeqnarraynumspace\nonumber\\*
&& \hfill \spaceland \bigl(l_{i,5} &\le& Q_Y(2) &\le& u_{i,5}\bigr)\IEEEeqnarraynumspace\nonumber\\*
&& \hfill \spaceland \bigl(l_{i,6} &\le& Q_Y(3) &\le& u_{i,6}\bigr)\bigr\},\IEEEeqnarraynumspace
\end{IEEEeqnarray}
where $l_{i,1},\ldots,l_{i,6}$ and $u_{i,1},\ldots,u_{i,6}$ are nonnegative numbers.
With the help of Lemma~\ref{lma:qilowerbound} below, we can verify that for specific $\rateqxqy \in \reals$ and $\Gamma \in \reals$ and for all $i \in \{1,\ldots,k\}$,
\begin{IEEEeqnarray}{l}
\infvphantomsup_{Q_{XY} \in \set{Q}_i} \hspace{0.1em} \sup_{\alpha \in (0,1]} \frac{1 - \alpha}{\alpha} \bigl(D_\alpha(P_{XY}\|Q_{XY}) - \rateqxqy\bigr) \ge \Gamma,\IEEEeqnarraynumspace
\end{IEEEeqnarray}
which by Lemma~\ref{lma:epeqinfsup} below and \eqref{eq:qunionqi} implies
\begin{IEEEeqnarray}{l}
\ratepxy(\rateqxqy) \ge \Gamma.\IEEEeqnarraynumspace
\end{IEEEeqnarray}
More details are given in Remark~\ref{rmk:epeqlowerbound}.

\begin{Lemma}
\label{lma:epeqinfsup}
For all $\rateqxqy \in \reals$,
\begin{IEEEeqnarray}{l}
\ratepxy(\rateqxqy) = \hspace{-0.1em} \infvphantomsup_{Q_{XY} \in \set{Q}} \hspace{0.1em} \sup_{\alpha \in (0,1]} \frac{1 - \alpha}{\alpha} \bigl(D_\alpha(P_{XY}\|Q_{XY}) - \rateqxqy\bigr)\hspace{-0.1em}.\IEEEeqnarraynumspace
\end{IEEEeqnarray}
\end{Lemma}

\begin{proof}
We have
\begin{IEEEeqnarray}{rCl}
\ratepxy(\rateqxqy) &=& \inf_{\substack{R_{XY} \in \set{P}(\set{X} \times \set{Y}):\\D(R_{XY}\|R_X R_Y) \le \rateqxqy}} D(R_{XY}\|P_{XY})\IEEEeqnarraynumspace\label{eq:epeqwithsupa}\\
&=& \! \inf_{Q_{XY} \in \set{Q}} \inf_{\substack{R_{XY} \in \set{P}(\set{X} \times \set{Y}):\\D(R_{XY}\|Q_{XY}) \le \rateqxqy}} D(R_{XY}\|P_{XY})\IEEEeqnarraynumspace\label{eq:epeqwithsupb}\\
&=& \! \infvphantomsup_{Q_{XY} \in \set{Q}} \hspace{0.1em} \sup_{\alpha \in (0,1]} \frac{1 - \alpha}{\alpha} \bigl(D_\alpha(P_{XY}\|Q_{XY}) - \rateqxqy\bigr)\hspace{-0.1em},\hspace{-0.1em}\IEEEeqnarraynumspace\label{eq:epeqwithsupc}
\end{IEEEeqnarray}
where \eqref{eq:epeqwithsupa} follows from \eqref{eq:defepofeq};
\eqref{eq:epeqwithsupb} follows from Lemma~\ref{lma:drxyrxryledrxyqxqy}; and
\eqref{eq:epeqwithsupc} follows from Lemma~\ref{lma:infdrqdrpeqsupalpha}.
\end{proof}

\begin{Lemma}
\label{lma:qilowerbound}
Let $\set{Q}_i$ be defined as in \eqref{eq:qxyidef}, let $\alpha \in (0,1)$, and let $\beta\colon \set{X} \times \set{Y} \to \reals_{\ge 0}$.
Define
\begin{IEEEeqnarray}{l}
D \triangleq \inf_{Q_{XY} \in \set{Q}_i} \sum_{(x,y) \in \set{X} \times \set{Y}} \beta(x,y) \, Q_{XY}(x,y)^{1 - \alpha}.\IEEEeqnarraynumspace\label{eq:qibounddefd}
\end{IEEEeqnarray}
Then, for all $\rateqxqy \in \reals$,
\begin{IEEEeqnarray}{rCl}
\IEEEeqnarraymulticol{3}{l}{\infvphantomsup_{Q_{XY} \in \set{Q}_i} \hspace{0.1em} \sup_{\tilde{\alpha} \in (0,1]} \frac{1 - \tilde{\alpha}}{\tilde{\alpha}} \bigl(D_{\tilde{\alpha}}(P_{XY}\|Q_{XY}) - \rateqxqy\bigr)}\IEEEeqnarraynumspace\nonumber\\*\quad
&\ge& \frac{-1}{\alpha} \log \Biggl\{\Biggl[\,\sum_{(x,y) \in \set{X} \times \set{Y}} \bigl(P_{XY}(x,y)^{\alpha} + \beta(x,y)\bigr)^\frac{1}{\alpha}\Biggr]^{\alpha}\IEEEeqnarraynumspace\nonumber\\*[0.2ex]
&& \hphantom{\frac{-1}{\alpha} \log \Biggl\{\binaryemptyleft} - D\Biggr\} - \frac{1 - \alpha}{\alpha} \, \rateqxqy.\IEEEeqnarraynumspace\label{eq:lmaqilowerbound}
\end{IEEEeqnarray}
\end{Lemma}

\begin{proof}
In the \hyperlink{prf:qilowerbound}{Appendix}.
\end{proof}

\pagebreak

\begin{Remark}
\label{rmk:epeqlowerbound}
We finish with a few comments about the verification of Example~\ref{exa:nonconvex}.
\begin{itemize}
\item Computing $D$ in Lemma~\ref{lma:qilowerbound} for fixed $\set{Q}_i$, $\alpha$, and $\beta$ is easy:
One can show that there exist an extreme point $Q_X^*$ of $\set{Q}_{X,i}$ and an extreme point $Q_Y^*$ of $\set{Q}_{Y,i}$ such that $Q_X^* Q_Y^*$ achieves the infimum in the RHS of \eqref{eq:qibounddefd}.
(This basically holds because $\set{Q}_{X,i}$ and $\set{Q}_{Y,i}$ are bounded convex polytopes and because the objective function is concave in $Q_X$ for fixed $Q_Y$ and concave in $Q_Y$ for fixed $Q_X$.)
Since $\set{Q}_{X,i}$ and $\set{Q}_{Y,i}$ have at most six extreme points, there are at most 36 candidate points.
One can evaluate the objective function at the candidate points; the minimum function value among these is equal to $D$.
\item To establish \eqref{eq:nonconvexubc}, we use \eqref{eq:qunionqi} with $k = 1\,323\,238$.
To ensure that \eqref{eq:qunionqi} holds, we start with a collection $\set{C}$ of sets that initially contains only $\set{Q}$;
we iteratively remove a $\set{Q}_i$ from $\set{C}$, split it into two parts, and add each part to $\set{C}$;
and we stop when $\set{C}$ has the desired structure.
\item We use interval arithmetic \cite{MPFI} to obtain exact bounds.
\item The splits to obtain $\set{C}$, the $\alpha$ and $\beta$ needed in Lemma~\ref{lma:qilowerbound} for every $i \in \{1,\ldots,k\}$, and the code that allows for a mathematically rigorous verification of our bounds can be found in \cite{github}.
\end{itemize}
\end{Remark}

\appendix
\label{sec:appendix}

\begin{proof}[Proof of Lemma~\ref{lma:achievablesanov}]
\hypertarget{prf:achievablesanov}{}
We skip the proofs of \eqref{eq:achsanova} and \eqref{eq:achsanovb}, which are similar to those of \eqref{eq:achsanovc} and \eqref{eq:achsanovd}.
For fixed $Q_X \in \set{P}(\set{X})$ and $Q_Y \in \set{P}(\set{Y})$, \eqref{eq:achsanovc} holds because
\begin{IEEEeqnarray}{rCl}
(Q_X Q_Y)^{\times n}[D(\hat{R}_{XY}\|P_{XY}) < \ratepxy + \epsilon] &\le& \tau \, e^{-n A}\IEEEeqnarraynumspace\label{eq:qxqysanovepepsilona}\\
&\le& \tau \, e^{-n (\rateqxqy + \epsilon)},\IEEEeqnarraynumspace
\end{IEEEeqnarray}
where \eqref{eq:qxqysanovepepsilona} follows from Sanov's theorem \cite[Theorem~11.4.1]{CoverThomas} with
\begin{IEEEeqnarray}{rCl}
A &\triangleq& \inf_{\substack{R_{XY} \in \set{P}(\set{X} \times \set{Y}):\\D(R_{XY}\|P_{XY}) < \ratepxy + \epsilon}} D(R_{XY}\|Q_X Q_Y)\IEEEeqnarraynumspace\\
&\ge& \inf_{\substack{R_{XY} \in \set{P}(\set{X} \times \set{Y}):\\D(R_{XY}\|P_{XY}) < \ratepxy + \epsilon}} D(R_{XY}\|R_X R_Y)\IEEEeqnarraynumspace\label{eq:qxqysanovepepsilond}\\
&\ge& \rateqxqy + \epsilon,\IEEEeqnarraynumspace\label{eq:qxqysanovepepsilone}
\end{IEEEeqnarray}
where \eqref{eq:qxqysanovepepsilond} follows from Lemma~\ref{lma:drxyrxryledrxyqxqy}, and
\eqref{eq:qxqysanovepepsilone} holds because $D(R_{XY}\|P_{XY}) < \ratepxy + \epsilon$ implies $D(R_{XY}\|R_X R_Y) \ge \rateqxqy + \epsilon$ by \eqref{eq:lmarxyepsepsilon}.
Similarly, for fixed $Q_X \in \set{P}(\set{X})$ and $Q_Y \in \set{P}(\set{Y})$, \eqref{eq:achsanovd} holds because
\begin{IEEEeqnarray}{rCl}
(Q_X Q_Y)^{\times n}[D(\hat{R}_{XY}\|\hat{R}_X \hat{R}_Y) \ge \rateqxqy + \epsilon] &\le& \tau \, e^{-n B}\IEEEeqnarraynumspace\label{eq:qxqysanovepepsilonk}\\
&\le& \tau \, e^{-n (\rateqxqy + \epsilon)},\hspace{-0.09em}\IEEEeqnarraynumspace
\end{IEEEeqnarray}
where \eqref{eq:qxqysanovepepsilonk} follows from Sanov's theorem with
\begin{IEEEeqnarray}{rCl}
B &\triangleq& \inf_{\substack{R_{XY} \in \set{P}(\set{X} \times \set{Y}):\\D(R_{XY}\|R_X R_Y) \ge \rateqxqy + \epsilon}} D(R_{XY}\|Q_X Q_Y)\IEEEeqnarraynumspace\\
&\ge& \inf_{\substack{R_{XY} \in \set{P}(\set{X} \times \set{Y}):\\D(R_{XY}\|R_X R_Y) \ge \rateqxqy + \epsilon}} D(R_{XY}\|R_X R_Y)\IEEEeqnarraynumspace\label{eq:qxqysanovepepsilonn}\\
&\ge& \rateqxqy + \epsilon,\IEEEeqnarraynumspace
\end{IEEEeqnarray}
where \eqref{eq:qxqysanovepepsilonn} follows from Lemma~\ref{lma:drxyrxryledrxyqxqy}.
\end{proof}

\begin{proof}[Proof of Lemma~\ref{lma:exponentnotachievable}]
\hypertarget{prf:exponentnotachievable}{}
Let $R_{XY}^{(1)}, R_{XY}^{(2)}, \ldots$ be a sequence of types and let $\delta_1, \delta_2, \ldots$ be a sequence of numbers with $\lim_{n \to \infty} \delta_n = 0$ such that for all $n \in \{1, 2, \ldots\}$, the following two inequalities are satisfied:
\begin{IEEEeqnarray}{rCl}
D\bigl(R_{XY}^{(n)}\big\|P_{XY}\bigr) &\le& \ratepxy + \delta_n,\IEEEeqnarraynumspace\\
D\bigl(R_{XY}^{(n)}\big\|R_X^{(n)} R_Y^{(n)}\bigr) &\le& \rateqxqy + \delta_n.\IEEEeqnarraynumspace\label{eq:rxyseqclosetoqxqy}
\end{IEEEeqnarray}
(The existence of such sequences follows from \eqref{eq:rxystarclosetopxyandrxry} and a continuity argument.)
Fix a sequence of tests $\{T_n\}_{n=1}^\infty$, and consider the function $f\colon \{1, 2, \ldots\} \to \reals$,
\begin{IEEEeqnarray}{l}
n \mapsto \min \Bigl\{ -\frac{1}{n} \log P_{XY}^{\times n}[T_n(X^n,Y^n) = 1] - \ratepxy,\IEEEeqnarraynumspace\nonumber\\*
\hspace{1.4em} \inf_{\qxcommaqylessspace} -\frac{1}{n} \log (Q_X Q_Y)^{\times n}[T_n(X^n,Y^n) = 0] - \rateqxqy \Bigr\}. \,\IEEEeqnarraynumspace
\end{IEEEeqnarray}
For a fixed $n \in \{1, 2, \ldots\}$, suppose that $\{(X_i,Y_i)\}_{i=1}^n$ are drawn uniformly at random from the type class corresponding to $R_{XY}^{(n)}$.
Observe that
\begin{IEEEeqnarray}{C}
\max\bigl\{\pr{T_n(X^n,Y^n) = 0}, \pr{T_n(X^n,Y^n) = 1}\bigr\} \ge \frac{1}{2}.\IEEEeqnarraynumspace\label{eq:prtngeoh}
\end{IEEEeqnarray}
If $\pr{T_n(X^n,Y^n) = 0} \ge \frac{1}{2}$, then
\begin{IEEEeqnarray}{rCl}
\IEEEeqnarraymulticol{3}{l}{(R_X^{(n)} R_Y^{(n)})^{\times n}[T_n(X^n,Y^n) = 0]}\IEEEeqnarraynumspace\nonumber\\*\qquad
&\ge& \frac{1}{2} \cdot \frac{1}{(n + 1)^{\card{\set{X} \times \set{Y}}}} \, e^{-n D(R_{XY}^{(n)}\|R_X^{(n)} R_Y^{(n)})},\IEEEeqnarraynumspace\label{eq:pxytnbadonehalf}
\end{IEEEeqnarray}
where \eqref{eq:pxytnbadonehalf} follows from \cite[Lemma~2.6]{CsiszarKorner}.
Thus,
\begin{IEEEeqnarray}{rCl}
\IEEEeqnarraymulticol{3}{l}{-\frac{1}{n} \log (R_X^{(n)} R_Y^{(n)})^{\times n}[T_n(X^n,Y^n) = 0] - \rateqxqy}\IEEEeqnarraynumspace\nonumber\\*\quad
&\le& \frac{\log \bigl[2 \,(n + 1)^{\card{\set{X} \times \set{Y}}}\bigr]}{n} + D\bigl(R_{XY}^{(n)}\big\|R_X^{(n)} R_Y^{(n)}\bigr) - \rateqxqy.\IEEEeqnarraynumspace
\end{IEEEeqnarray}
Together with \eqref{eq:rxyseqclosetoqxqy}, this implies
\begin{IEEEeqnarray}{rCl}
f(n) \le \frac{\log \bigl[2 \,(n + 1)^{\card{\set{X} \times \set{Y}}}\bigr]}{n} + \delta_n.\IEEEeqnarraynumspace\label{eq:minerrorexpupperbound}
\end{IEEEeqnarray}
Using similar arguments, it can be shown that \eqref{eq:minerrorexpupperbound} also holds if $\pr{T_n(X^n,Y^n) = 1} \ge \frac{1}{2}$.
We conclude that for any sequence of tests $\{T_n\}_{n=1}^\infty$, $\limsup_{n \to \infty} f(n) \le 0$, which implies that \eqref{eq:deferrexponentpxy} and \eqref{eq:deferrexponentqxqy} cannot hold simultaneously.
(All steps of the proof remain valid when randomized tests are allowed.)
\end{proof}

\begin{proof}[Proof of Lemma~\ref{lma:qilowerbound}]
\hypertarget{prf:qilowerbound}{}
Observe that
\begin{IEEEeqnarray}{rCl}
\IEEEeqnarraymulticol{3}{l}{\infvphantomsup_{Q_{XY} \in \set{Q}_i} \hspace{0.1em} \sup_{\tilde{\alpha} \in (0,1]} \frac{1 - \tilde{\alpha}}{\tilde{\alpha}} \bigl(D_{\tilde{\alpha}}(P_{XY}\|Q_{XY}) - \rateqxqy\bigr)}\IEEEeqnarraynumspace\nonumber\\*\quad
&\ge& \sup_{\tilde{\alpha} \in (0,1]} \hspace{0.1em} \infvphantomsup_{Q_{XY} \in \set{Q}_i} \frac{1 - \tilde{\alpha}}{\tilde{\alpha}} \bigl(D_{\tilde{\alpha}}(P_{XY}\|Q_{XY}) - \rateqxqy\bigr)\IEEEeqnarraynumspace\label{eq:qisupinfa}\\
&\ge& \frac{1-\alpha}{\alpha} \left[\inf_{Q_{XY} \in \set{Q}_i} D_\alpha(P_{XY}\|Q_{XY}) - \rateqxqy\right]\!,\IEEEeqnarraynumspace\label{eq:qisupinfb}
\end{IEEEeqnarray}
where \eqref{eq:qisupinfa} follows from the minimax inequality \cite[(2.28)]{Bertsekas}, and
\eqref{eq:qisupinfb} holds because $\alpha \in (0,1)$ and $\frac{1-\alpha}{\alpha} \ge 0$.
Moreover,
\begin{IEEEeqnarray}{rCl}
\IEEEeqnarraymulticol{3}{l}{\inf_{Q_{XY} \in \set{Q}_i} D_\alpha(P_{XY}\|Q_{XY})}\IEEEeqnarraynumspace\nonumber\\*\quad
&=& \frac{1}{\alpha-1} \log \sup_{Q \in \set{Q}_i} \sum_{(x,y) \in \set{X} \times \set{Y}} P_{XY}(x,y)^\alpha \, Q(x,y)^{1-\alpha},\IEEEeqnarraynumspace\label{eq:qiinfsuma}
\end{IEEEeqnarray}
where \eqref{eq:qiinfsuma} holds because $\frac{1}{\alpha-1} < 0$ and because the logarithm is an increasing function.
Next, using the shorthand notations $z \triangleq (x,y)$, $\set{Z} \triangleq \set{X} \times \set{Y}$, and $P \triangleq P_{XY}$,
\begin{IEEEeqnarray}{rCll}
\IEEEeqnarraymulticol{4}{l}{\sup_{Q \in \set{Q}_i} \sum_{z \in \set{Z}} P(z)^\alpha \, Q(z)^{1-\alpha}}\IEEEeqnarraynumspace\nonumber\\*
\quad &=& \sup_{Q \in \set{Q}_i} \Biggl[ \, & \sum_{z \in \set{Z}} P(z)^\alpha \, Q(z)^{1-\alpha} + \sum_{z \in \set{Z}} \beta(z) \, Q(z)^{1-\alpha}\IEEEeqnarraynumspace\nonumber\\*[-0.5ex]
&&& -\> \sum_{z \in \set{Z}} \beta(z) \, Q(z)^{1-\alpha}\Biggr]\IEEEeqnarraynumspace\label{eq:supqibounda}\\
&=& \sup_{Q \in \set{Q}_i} \Biggl[ \, & \sum_{z \in \set{Z}} \bigl(P(z)^\alpha + \beta(z)\bigr) \, Q(z)^{1-\alpha}\IEEEeqnarraynumspace\nonumber\\*[-0.5ex]
&&& -\> \sum_{z \in \set{Z}} \beta(z) \, Q(z)^{1-\alpha}\Biggr]\IEEEeqnarraynumspace\label{eq:supqiboundb}\\
&\le& \IEEEeqnarraymulticol{2}{l}{\sup_{Q \in \set{Q}_i} \sum_{z \in \set{Z}} \bigl(P(z)^\alpha + \beta(z)\bigr) \, Q(z)^{1-\alpha}}\IEEEeqnarraynumspace\nonumber\\*
&& \IEEEeqnarraymulticol{2}{l}{-\> \inf_{Q \in \set{Q}_i} \sum_{z \in \set{Z}} \beta(z) \, Q(z)^{1-\alpha}}\IEEEeqnarraynumspace\label{eq:supqiboundc}\\
&=& \IEEEeqnarraymulticol{2}{l}{\sup_{Q \in \set{Q}_i} \sum_{z \in \set{Z}} \bigl(P(z)^\alpha + \beta(z)\bigr) \, Q(z)^{1-\alpha} - D}\IEEEeqnarraynumspace\label{eq:supqiboundd}\\
&\le& \IEEEeqnarraymulticol{2}{l}{\sup_{Q \in \set{Q}_i} \Biggl\{\Biggl[\,\sum_{z \in \set{Z}} \bigl(P(z)^\alpha + \beta(z)\bigr)^{\frac{1}{\alpha}}\Biggr]^{\alpha}}\IEEEeqnarraynumspace\nonumber\\*[-0.5ex]
&& \IEEEeqnarraymulticol{2}{l}{\hphantom{\sup_{Q \in \set{Q}_i} \Biggl\{} \cdot \!\Biggl[\,\sum_{z \in \set{Z}} Q(z)\Biggr]^{1-\alpha} \Biggr\} - D}\IEEEeqnarraynumspace\label{eq:supqibounde}\\[0.5ex]
&=& \IEEEeqnarraymulticol{2}{l}{\Biggl[\,\sum_{z \in \set{Z}} \bigl(P(z)^\alpha + \beta(z)\bigr)^{\frac{1}{\alpha}}\Biggr]^{\alpha} - D,}\IEEEeqnarraynumspace\label{eq:supqiboundf}
\end{IEEEeqnarray}
where \eqref{eq:supqiboundc} holds because the supremum of a sum is upper bounded by the sum of the suprema;
\eqref{eq:supqiboundd} follows from the definition of $D$;
\eqref{eq:supqibounde} follows from H\"older's inequality; and
\eqref{eq:supqiboundf} holds because $\sum_{z \in \set{Z}} Q(z) = 1$.
Now, \eqref{eq:lmaqilowerbound} follows from combining \eqref{eq:qisupinfb}, \eqref{eq:qiinfsuma}, and \eqref{eq:supqiboundf}.
\end{proof}

\end{document}